\documentclass[12pt]{llncs}

\usepackage{times}

\usepackage{mathtools} 
\usepackage{changepage}
\usepackage{enumitem} 
\usepackage{xcolor} 
\usepackage{amssymb}              
\usepackage{caption} 

\newcommand{\pve}{\textsf{pos}}
\newcommand{\nve}{\textsf{neg}}

\newcommand{\Ainit}{\ensuremath{A_{\textbf{0}}}}   
\newcommand{\Rp}{\ensuremath{R_{\mathbf{p}}}}

\newcommand{\APA}{\textsf{APA}}

\newcommand{\qinit}{\ensuremath{q_{\textbf{0}}}}
\newcommand{\qend}{\ensuremath{q_{\textbf{f}}}}

\newcommand{\hide}[1]{} 
\newcommand{\lol}[1]{\ensuremath{\multimap^{\!\!\!\!\!\!\!\!\!#1}}}  

\newcommand{\andC}{\textsf{and}}

\begin{document}
     \title{Turing-Completeness of Dynamics in \\ Abstract Persuasion Argumentation} 

\author{Ryuta Arisaka
    \institute{ryutaarisaka@gmail.com} 
    }

    \maketitle
\begin{abstract} 
   Abstract Persuasion Argumentation (\APA) 
   is a dynamic argumentation formalism
   that extends Dung argumentation 
   with persuasion relations. 
   In this work, we show through two-counter Minsky machine encoding that $\APA$ 
   dynamics is Turing-complete. 
\end{abstract} 
\section{Introduction}       
Abstract Persuasion Argumentation ($\APA$) \cite{ArisakaSatoh18} 
is a dynamic argumentation formalism that extends Dung argumentation \cite{Dung95}
with persuasion relations.    Not only can an argument in $\APA$ attack an argument just in Dung argumentation, 
   it can also induce an argument, or convert an argument into another argument. 
   Dung argumentation is a state in $\APA$, which 
   the persuasion relations may modify into other states.  Transitions 
at each state are effected 
with respect to a selected subset of the arguments in the state. The subset - termed a reference set \cite{ArisakaSatoh18} - 
provides defence against external persuaders just as against external attackers. 
From 
the set of all the dynamic modifications by external 
persuaders the reference set does not defend against, 
any number of them may be executed simultaneously into another state.  

In this work, we study the computational 
capability of $\APA$ dynamics, 
its relation to Turing machines, specifically. The investigation 
is useful: as $\APA$ is a conservative extension of Dung argumentation, the persuasions 
specified similarly to attacks, 
Turing-completeness of $\APA$ dynamics induced by persuasions would mean that 
drastic departure from the intuitive Dung-based theoretical framework for dealing with dynamics
would not be necessary.
We settle this research problem positively through 
two-counter Minsky machine \cite{Minsky67} encoding.

\section{Technical Backgrounds} 
\textbf{Two-counter Minsky machines.} Let $\mathbb{N}$ be the class of natural numbers including 0, whose 
member is referred to by $n$ with or without a subscript
and a superscript, 
and let $Q$ be a finite set of abstract entities called states, whose 
member is referred to by $q$ with or without a subscript. 
A two-counter (non-deterministic) Minsky machine \cite{Minsky67} can be defined to be
a tuple $(Q, n_1, n_2, \Rightarrow, \qinit, \qend, \mathcal{I})$ 
with: $n_1, n_2 \in \mathbb{N}$; $\Rightarrow: (Q \backslash \ \{\qend\}) \times \mathbb{N} 
\times \mathbb{N} \rightarrow Q \times \mathbb{N} \times \mathbb{N}$; 
and $\mathcal{I}: (Q \backslash \{\qend\}) \times \{1,2\}  \times Q \times (Q \cup \{\epsilon\})$, with
$\epsilon \not\in Q$.
$\qinit$ is called the initial state, 
and $\qend$ is called the halting state. Each member of $\mathcal{I}$ is called an 
instruction. For every $q 
\in (Q \backslash \{\qend\})$, 
there is some $(q, i, q_2, u) \in \mathcal{I}$ 
for $i \in \{1,2\}$, $q_2 \in Q$, $u \in Q \cup \{\epsilon\}$. 
$\Rightarrow$ is specifically defined to be such that 
$\Rightarrow((q_1, n_1, n_2))$ is:
\begin{itemize}[leftmargin=0cm]
    \item[] $(q_2, n_1 + 1, n_2)$ if $(q_1, 1, q_2, \epsilon) \in \mathcal{I}$.
    \item[]
    $(q_2, n_1, n_2 + 1)$ if $(q_1, 2, q_2, \epsilon) \in \mathcal{I}$.
    \item[] $(q_3, n_1 - 1, n_2)$ if $(q_1, 1, q_2, q_3) \in \mathcal{I}$ $\andC$ $n_1 > 0$.\footnote{``$\andC$'' instead of 
 ``and'' is used in this and all the other papers to be written by 
 the author when the context in which the word appears 
 strongly indicates truth-value comparisons. It follows the semantics 
 of classical logic conjunction.}  
    \item[]
    $(q_3, n_1, n_2 -1)$ if $(q_1, 2, q_2, q_3) \in \mathcal{I}$ 
    $\andC$ $n_2 > 0$.
     \item[] $(q_2, n_1, n_2)$\qquad if $(q_1, 1, q_2, q_3) \in \mathcal{I}$ $\andC$ $n_1 = 0$. 
    \item[]
    $(q_2, n_1, n_2)$\qquad if $(q_1, 2, q_2, q_3) \in \mathcal{I}$ 
    $\andC$ $n_2 = 0$.
\end{itemize}
Minsky machine is said to halt on $(n_1, n_2)$ iff there are 
some $n_x, n'_1, n'_2 \in \mathbb{N}$ such that $\Rightarrow^{n_x}((\qinit, n_1, n_2)) = 
(\qend, n_1', n_2')$. 

\noindent \textbf{$\APA$: Abstract Persuasion Argumentation.}
Let $\mathcal{A}$ be a class of abstract entities that we understand 
as arguments, whose member is referred to by $a$ with or without a subscript and a superscript, 
and whose subset is referred to by $A$ with or without a subscript and a superscript. 
$\APA$  \cite{ArisakaSatoh18} is a tuple  
$(A, R, \Rp, \Ainit, \hookrightarrow)$ 
for $\Ainit \subseteq A$; $R: A \times A$; 
$\Rp:  
 A \times (A \cup \{\epsilon\}) \times 
A$; and 
 $\hookrightarrow: 2^A \times (2^{(A, R)} \times 2^{(A, R)})$. It extends 
 Dung argumentation $(A, R)$ \cite{Dung95} conservatively. 
$(a_1, a_2) \in R$
is drawn graphically as $a_1 \rightarrow a_2$; 
$(a_1, \epsilon, a_2) \in \Rp$ 
is drawn graphically as $a_1 \multimap a_2$;
and $(a_1, a_3, a_2) \in \Rp$ is drawn graphically as 
$a_1 \dashrightarrow a_3 \lol{\!a_1}\ a_2$.

Let $F(A_1)$ for some $A_1 \subseteq A$
denote $(A_1, R \cap (A_1 \times A_1))$, 
$F(A_1)$ is said to 
be a state. $F(\Ainit)$ is called the initial state in particular. 
In any state $F(A_x)$, any member of $A_x$ is said to be visible 
in $F(A_x)$, while the others are said to be invisible in $F(A_x)$. 
A state $F(A_1)$ is said to be reachable iff $F(A_1) = F(\Ainit)$ or else there is some 
$F(A_2)$ such that $F(A_2)$ is reachable and that $F(A_2) \hookrightarrow^{A_x} F(A_1)$ 
for some $A_x \subseteq A$ called a reference set for the transition.

 $a_1 \in A$ is said to attack $a_2 \in A$ in a state $F(A_1)$ iff 
 $a_1, a_2 \in A_1$ $\andC$ $(a_1, a_2) \in R$.
 For $a_1, a_2, a_3 \in A$, 
 $a_1$ is said to be: inducing $a_2$ in a state $F(A_1)$ with respect to a reference set 
 $A_x \subseteq A$ iff 
 $a_1 \in A_1$ $\andC$
 $(a_1, \epsilon, a_2) \in \Rp$ $\andC$ 
 $a_1$ is not attacked by any member of $A_x$ in $F(A_1)$; 
 and converting $a_3$ into $a_2$ in a state $F(A_1)$ with respect to a reference set 
 $A_x \subseteq A$ iff $a_1, a_3 \in A_1$ $\andC$
 $(a_1, a_3, a_2) \in \Rp$ $\andC$ 
 $a_1$ is not attacked by any member of $A_x$ in $F(A_1)$. 
The set of all members of 
$\Rp$ that are inducing or converting in $F(A_1)$ with respect to a reference set 
$A_x \subseteq A$ are denoted by $\Gamma^{A_x}_{F(A_1)}$.

\indent \textbf{Interpretation of $\hookrightarrow$.} The interpretation 
given of $\hookrightarrow$ in \cite{ArisakaSatoh18} is: (1) 
any subset $\Gamma$ of $\Gamma^{A_x}_{F(A_1)}$ can be simultaneously 
considered for transition for $F(A_1)$ into some $F(A_2)$; (2) 
if $\Gamma \subseteq \Gamma^{A_x}_{F(A_1)}$ is considered for transition 
into $F(A_2)$, then: (2a) if either $(a_1, \epsilon, a_2)$ or 
$(a_1, a_3, a_2)$ is in $\Gamma$, then 
$a_2 \in A_2$; (2b) if 
$(a_1, a_3, a_2)$ is in $\Gamma$, then 
$a_3$ is not in $A_2$ unless it is judged to be in $A_2$ by (2a); 
and (2c) if $a \in A_1$, then $a \in A_2$ unless 
it is judged not in $A_2$ by (2b). 

In other words, for $A_1 \subseteq A$  
and for $\Gamma \subseteq \Rp$,  
let $\nve^{A_1}(\Gamma)$ 
   be $\{a_x \in A_1 \ | \ \exists a_1, a_2 \in A_1.(a_1, a_x, a_2) 
    \in \Gamma\}$, 
   and let $\pve^{A_1}(\Gamma)$  
   be $\{a_2 \in A \ | \ \exists a_1, \alpha \in A_1 \cup 
   \{\epsilon\}.(a_1, \alpha, a_2)
   \in \Gamma\}$. 
   For $A_x \subseteq A$, $F(A_1)$ and $F(A_2)$, we have:
   $(A_x, F(A_1), F(A_2)) \in\ \hookrightarrow$, alternatively
   $F(A_1) \hookrightarrow^{A_x} F(A_2)$,  iff
   there is some $\emptyset \subset \Gamma \subseteq 
\Gamma^{A_x}_{F(A_1)} \subseteq \Rp$
   such that 
   $A_2 = (A_1 \backslash \nve^{A_1}(\Gamma))
    \cup \pve^{A_1}(\Gamma)$.

\textbf{State-wise acceptability semantics.} We touch upon state-wise $\APA$ acceptability semantics, only briefly, since they 
are not required in this work. $A_1 \subseteq A$ is said to be conflict-free 
in a (reachable) state $F(A_a)$ iff 
no member of $A_1$ attacks 
a member of $A_1$ in $F(A_a)$. 
$A_1 \subseteq A$ is said to defend $a \in A$
in $F(A_a)$ iff, if  $a \in A_a$, 
then both: (1) every $a_u \in A_a$ attacking $a$ 
in $F(A_a)$ is attacked 
by at least one member of $A_1$ in $F(A_a)$ (counter-attack); $\andC$ (2) there is no state $F(A_b)$ 
such that both $F(A_a) \hookrightarrow^{A_1} F(A_b)$ 
and $a \not\in A_b$ at once (no elimination). 
 $A_1 \subseteq A$ is said to be proper 
in $F(A_a)$ 
iff $A_1 \subseteq A_a$. 

$A_1 \subseteq A$ is said to be: admissible 
in $F(A_a)$ iff $A_1$ is conflict-free, proper and
defends every member of $A_1$ in $F(A_a)$; and
complete in $F(A_a)$ iff $A_1$ is admissible and 
includes every $a \in A$ it defends in $F(A_a)$.\\ 

\section{Encoding of Minsky Machines into $\APA$}
Assume a two-counter Minsky machine $(Q, n_1, n_2, \Rightarrow, 
\qinit, \qend, \mathcal{I})$. For the encoding into 
$\APA$, assume injective functions: $\sigma^1, \sigma^2: \mathbb{N} \rightarrow \mathcal{A}$;
$\sigma^Q: Q \rightarrow \mathcal{A}$; 
$\sigma^{\mathcal{I}}, \sigma^{\mathcal{I}c}: \mathcal{I} \rightarrow \mathcal{A}$,
such that for any two distinct 
$x, y \in \{1,2,Q,\mathcal{I},\mathcal{I}c\}$,
$\textsf{range}(\sigma^x) \cap \textsf{range}(\sigma^y)
= \emptyset$. Assume $A$ to be the set that satisfies all the following. $A$ is naturally unbounded. 
\begin{itemize}
    \item[] {\small $\sigma^1(n), \sigma^2(n) \in A$ iff 
        $n \in \mathbb{N}$. \quad 
   $\sigma^Q(q) \in A$ iff $q \in Q$. \quad 
     $\sigma^{\mathcal{I}}(x), \sigma^{\mathcal{I}c}(x) \in A$ iff $x \in \mathcal{I}$.}
\end{itemize}
We denote the subset of $A$ consisting of: all $\sigma^1(n)$ by $A^1$; all $\sigma^2(n)$ by $A^2$; 
all $\sigma^Q(q)$ by $A^Q$; all $\sigma^{\mathcal{I}}(x)$ by $A^{\mathcal{I}}$; 
and all $\sigma^{\mathcal{I}c}(x)$ by $A^{\mathcal{I}c}$. Assume $R = \emptyset$. Assume $\Rp$ as the set intersection of all
the sets that satisfy:  
\begin{enumerate}
    \item $(\sigma^{\mathcal{I}}(x), \sigma^{\mathcal{Q}}(q_1),  
            \sigma^{\mathcal{I}c}(x)) \in \Rp$ for every $x \equiv (q_1, i, q_2, u)$ 
            for some $i \in \{1,2\}$, 
            $q_1 \in (Q \backslash \{\qend\}), q_2 \in Q$, 
            $u \in Q \cup \{\epsilon\}$ and for every
            $n$.
    \item $(\sigma^{\mathcal{I}c}(x), 
          \sigma^i(n), \sigma^i(n+1))
          \in \Rp$ for every 
          $x \equiv (q_1, i, q_2, \epsilon) \in \mathcal{I}$ for some   
        $q_1 \in (Q \backslash \{\qend\}), q_2 \in Q$, $i \in \{1,2\}$, and for every $n$. 
    \item $(\sigma^{\mathcal{I}c}(x), 
          \sigma^i(n), \sigma^i(n-1))
          \in \Rp$ for every 
          $x \equiv (q_1, i, q_2, q_3) \in \mathcal{I}$ for some   
        $q_1 \in (Q \backslash \{\qend\}), q_2, q_3 \in Q$, $i \in \{1,2\}$, and for every $n \not= 0$. 
    \item $(\sigma^i(n), \sigma^{\mathcal{I}c}(x),
        \sigma^Q(q_2)) \in \Rp$ 
        for every $x \equiv (q_1, i, q_2, \epsilon) \in \mathcal{I}$ for 
        some $q_1 \in (Q \backslash \{\qend\})$, $q_2 \in Q$, 
        $i \in \{1,2\}$, and for every 
        $n$.
         \item $(\sigma^i(n), \sigma^{\mathcal{I}c}(x),
        \sigma^Q(q_3)) \in \Rp$ 
        for every $x \equiv (q_1, i, q_2, q_3) \in \mathcal{I}$ for 
        some $q_1 \in (Q \backslash \{\qend\})$, $q_2, q_3 \in Q$, 
        $i \in \{1,2\}$, and for every 
        $n \not= 0$.
    \item $(\sigma^i(0), \sigma^{\mathcal{I}c}(x),
        \sigma^Q(q_2)) \in \Rp$ 
        for every $x \equiv (q_1, i, q_2, q_3) \in \mathcal{I}$ for 
        some $q_1 \in (Q \backslash \{\qend\})$, $q_2, q_3 \in Q$, 
        $i \in \{1,2\}$.
\end{enumerate}
Assume $\Ainit$ to be $\{\sigma^1(n_1), \sigma^2(n_2), \sigma^Q(\qinit)\} \cup 
\bigcup_{x \in \mathcal{I}}\sigma^{\mathcal{I}}(x)$. 
\begin{theorem}[Turing-completeness]{\ }\\
Such $(A, R, \Rp, \Ainit, \hookrightarrow)$ simulates 
$(Q, n_1, n_2, \Rightarrow, \qinit, \qend, \mathcal{I})$. 
\end{theorem}
\begin{proof} 
By induction on $k$ of $\Rightarrow^k((q_1, n_1, n_2)) = (q_2, n'_1, n'_2)$, we show that 
there is a corresponding $\APA$ transition 
$F(A^{\mathcal{I}} \cup \sigma^Q(q_1) \cup 
\sigma^1(n_1) \cup 
\sigma^2(n_2)) \hookrightarrow^{A_x} 
\cdots \hookrightarrow^{A_x}
F(A^{\mathcal{I}} \cup \sigma^Q(q_2) \cup 
\sigma^1(n'_1) \cup \sigma^2(n'_2)
  )$. Since $R =   \emptyset$, it is not important what $A_x$ here is.

If there is no $0 < k$, then $q_1 = \qend$, 
and there is no $\APA$ state $F(A_2)$ such that 
$F(A^{\mathcal{I}} \cup \sigma^Q(\qend) \cup \sigma^1(n_1) \cup \sigma^2(n_2)))
\hookrightarrow^{A_x} F(A_2)$, for: (1) every 
$a \in A^{\mathcal{I}}$ is converting only members of $A^Q \backslash \{\sigma^Q(\qend)\}$ in any state; 
(2) $\sigma^i(n_i)$, $i \in \{1,2\}$, is not converting any member of $A^{\mathcal{I}}$, of $A^Q$, of $A^1$ or of $A^2$;
(3) no member of $A^Q$ is converting any member of $A$; and (4) every member of $\Rp$ 
is a conversion. 

For inductive cases, assume that the correspondence holds for any $k \leq j$. 
We show by cases that it holds for $k = j+1$ as well. 
\begin{center}
    \includegraphics[scale=0.09]{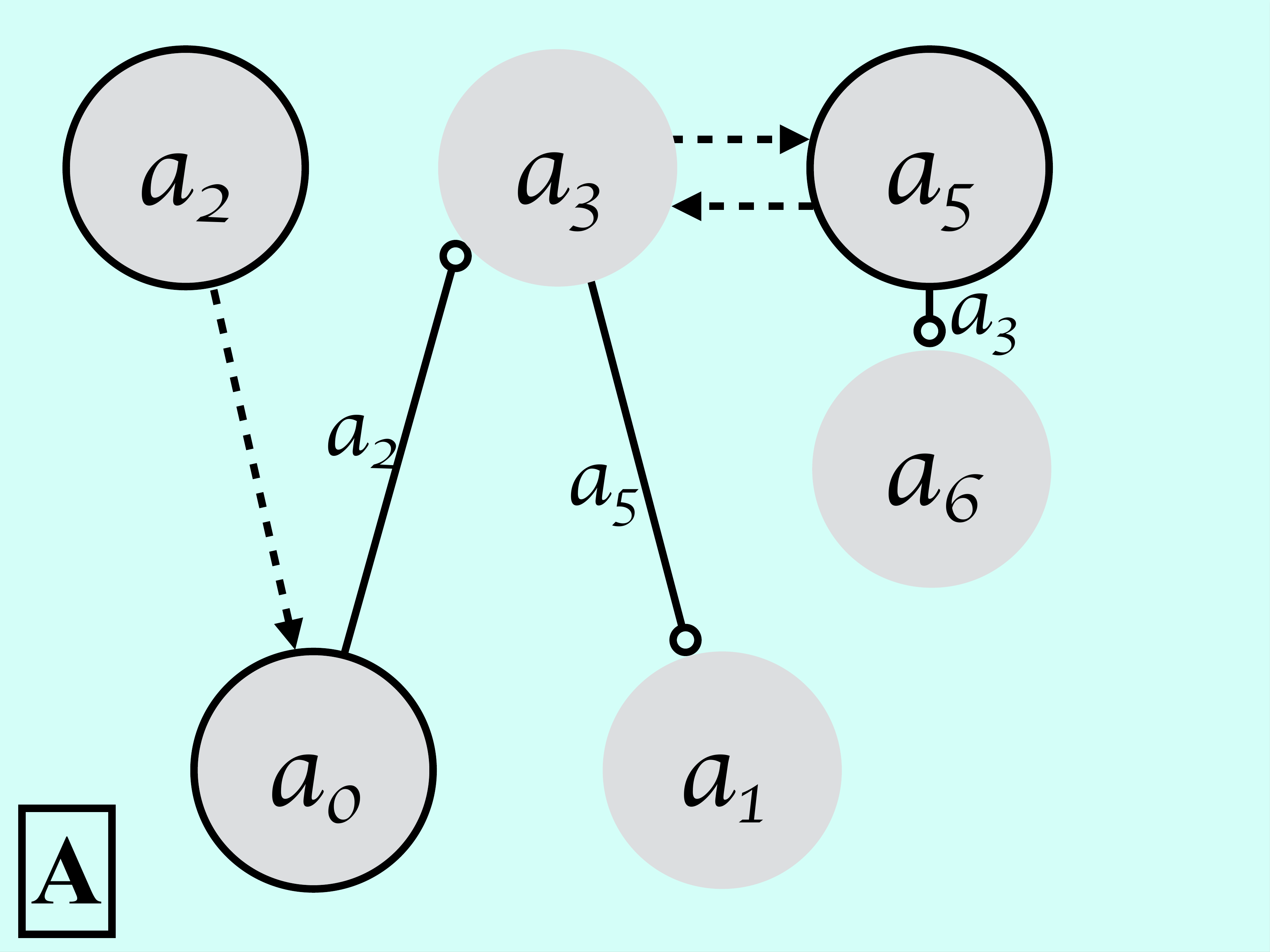}
    \includegraphics[scale=0.09]{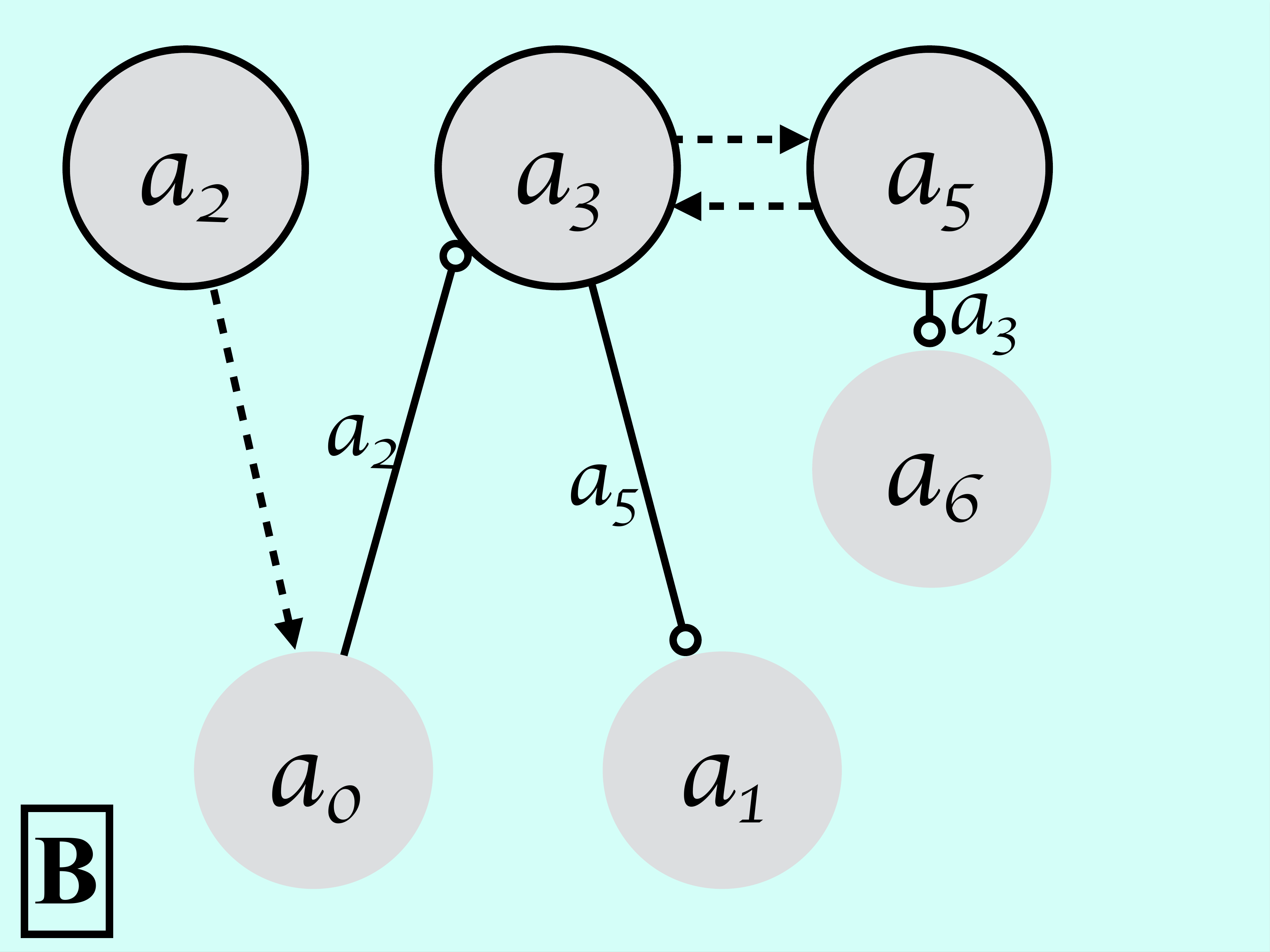}
    \includegraphics[scale=0.09]{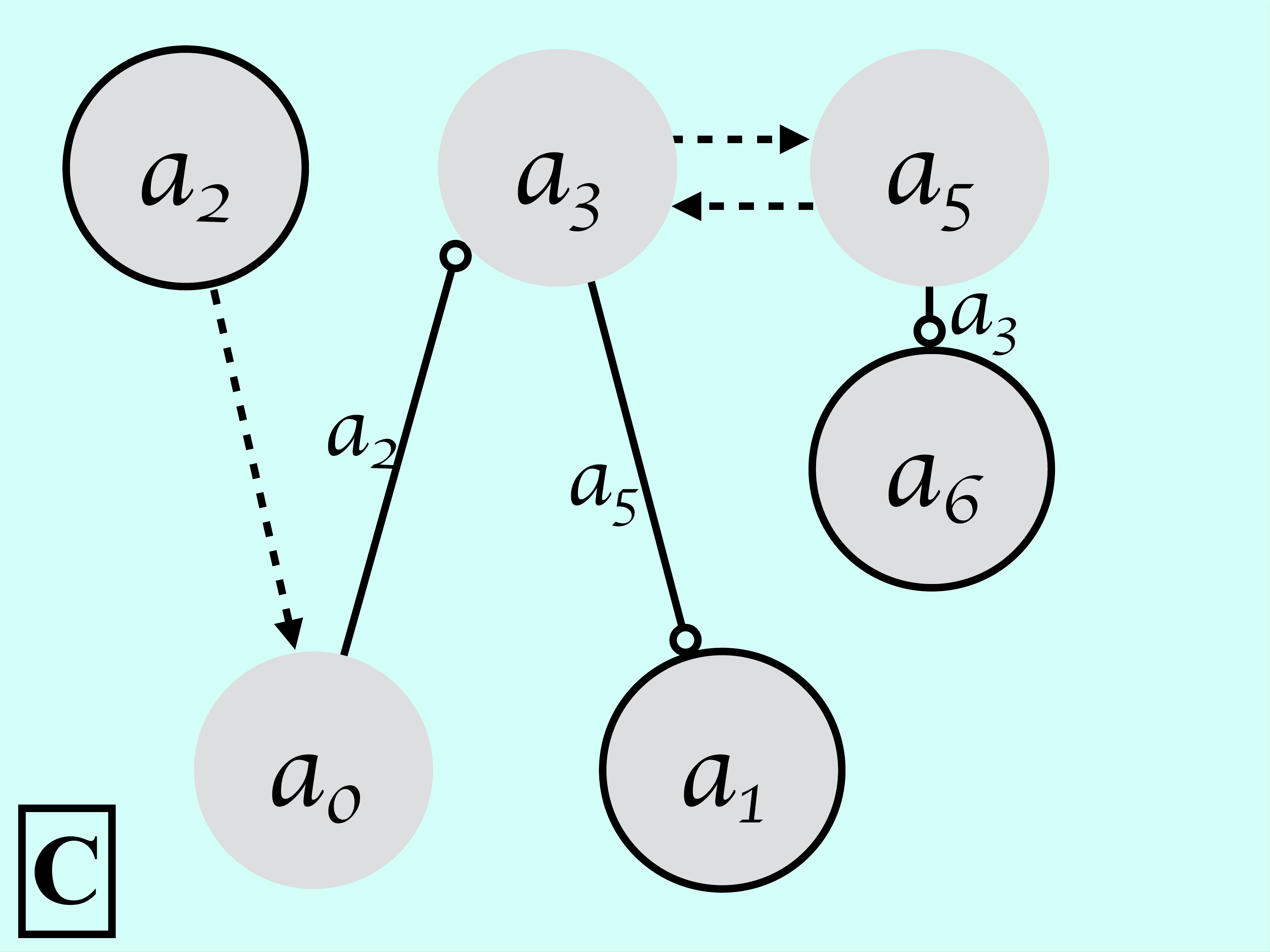}
\end{center}
\begin{description}
   \item[Case 1] $(q_1, n_1, n_2) \underbrace{\Rightarrow \cdots \Rightarrow}_j (q_{2}, n_3, n_4) 
   \Rightarrow (q_3, n_3 + 1, n_4)$: By induction hypothesis, we have 
   $F(A^{\mathcal{I}} \cup \sigma^Q(q_1) \cup \sigma^1(n_1) \cup \sigma^2(n_2))
   \hookrightarrow^{A_x} \cdots \hookrightarrow^{A_x} F(A^{\mathcal{I}} \cup \sigma^Q(q_2) \cup \sigma^1(n_3) \cup \sigma^2(n_4)) \equiv F(A_1)$. 
      
     {\ }\indent A relevant snippet of the $\APA$ for the $(j+1)$-th Minsky machine computation is shown in \fbox{A}. All visible arguments in $F(A_1)$ are bordered. 
       Assume $x \equiv (q_2, 1, q_3, \epsilon)$, 
       $a_0 = \sigma^Q(q_2), a_1 = \sigma^Q(q_3)$, $a_2 = \sigma^{\mathcal{I}}(x)$, 
       $a_3 = \sigma^{\mathcal{I}c}(x)$,  
       $a_5 = \sigma^1(n_1)$, and $a_6 = \sigma^1(n_1+ 1)$. 
       Then, 
       $a_0$ and $a_1$ represent states of Minsky machine, 
       $a_5$ and $a_6$ represent the first counter's content, 
       while $a_3$ is an auxiliary operational $\APA$ argument. 
       
       {\ }\indent By the construction of the $\APA$, we have $(a_2, a_0, a_3) \in \Rp$. 
       Furthermore, among all $a \in A^Q$, only $a_0$ is 
       in $F(A_1)$. There then 
       exists a transition: $F(A_1) \hookrightarrow^{A_x} 
       F((A_1 \backslash \{a_0\}) \cup \{a_3\}) \equiv F(A_2)$ (\fbox{B}). 
       
       {\ }\indent By the construction of the $\APA$, among all $a \in A^1$, 
       only $a_5$ is in $F(A_1)$, which is true also in 
       $F(A_2)$. Of all $\APA$ arguments 
       converting (converted by) $a_5$ in $F(A_2)$, only $a_3$ is in  $F(A_2)$. 
       There exists a $\APA$ transition 
       $F(A_2) \hookrightarrow^{A_x} F((A_1 \backslash \{a_3, a_5\}) \cup 
       \{a_1, a_6\}) \equiv F(A_3)$ (\fbox{C}). However, 
       $A_3 = A^{\mathcal{I}} \cup \sigma^Q(q_3) \cup \sigma^1(n_1+1) \cup \sigma^2(n_2)$, as required. 
  \item[Case 2]  $(q_1, n_1, n_2) \underbrace{\Rightarrow \cdots \Rightarrow}_j (q_{2}, n_3, n_4) 
   \Rightarrow (q_3, n_3, n_4 + 1)$: similar. 
   \item[Case 3] $(q_1, n_1, n_2) \underbrace{\Rightarrow \cdots \Rightarrow}_j (q_{2}, n_3, n_4) 
   \Rightarrow (q_4, n_3 - 1, n_4)$: By induction hypothesis, we have 
   $F(A^{\mathcal{I}} \cup \sigma^Q(q_1) \cup \sigma^1(n_1) \cup \sigma^2(n_2))
   \hookrightarrow^{A_x} \cdots \hookrightarrow^{A_x} F(A^{\mathcal{I}} \cup \sigma^Q(q_2) \cup \sigma^1(n_3) \cup \sigma^2(n_4)) \equiv F(A_1)$. 
      {\ }\indent A relevant snippet of the $\APA$ for the $(j+ 1)$-th Minsky machine  computation  is  as  shown  in   \fbox{D}.  
           
\begin{center}
    \includegraphics[scale=0.10]{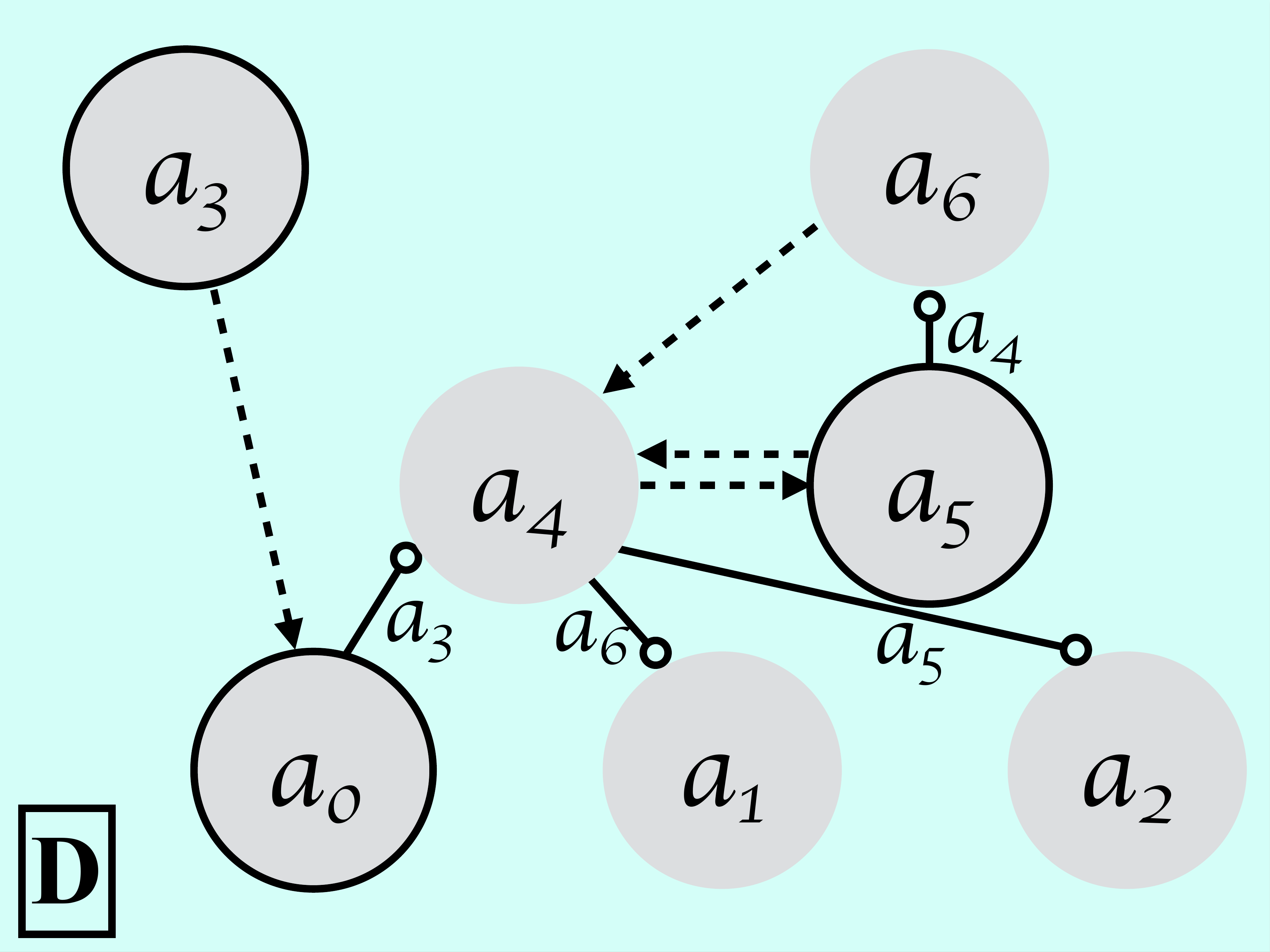}
    \includegraphics[scale=0.10]{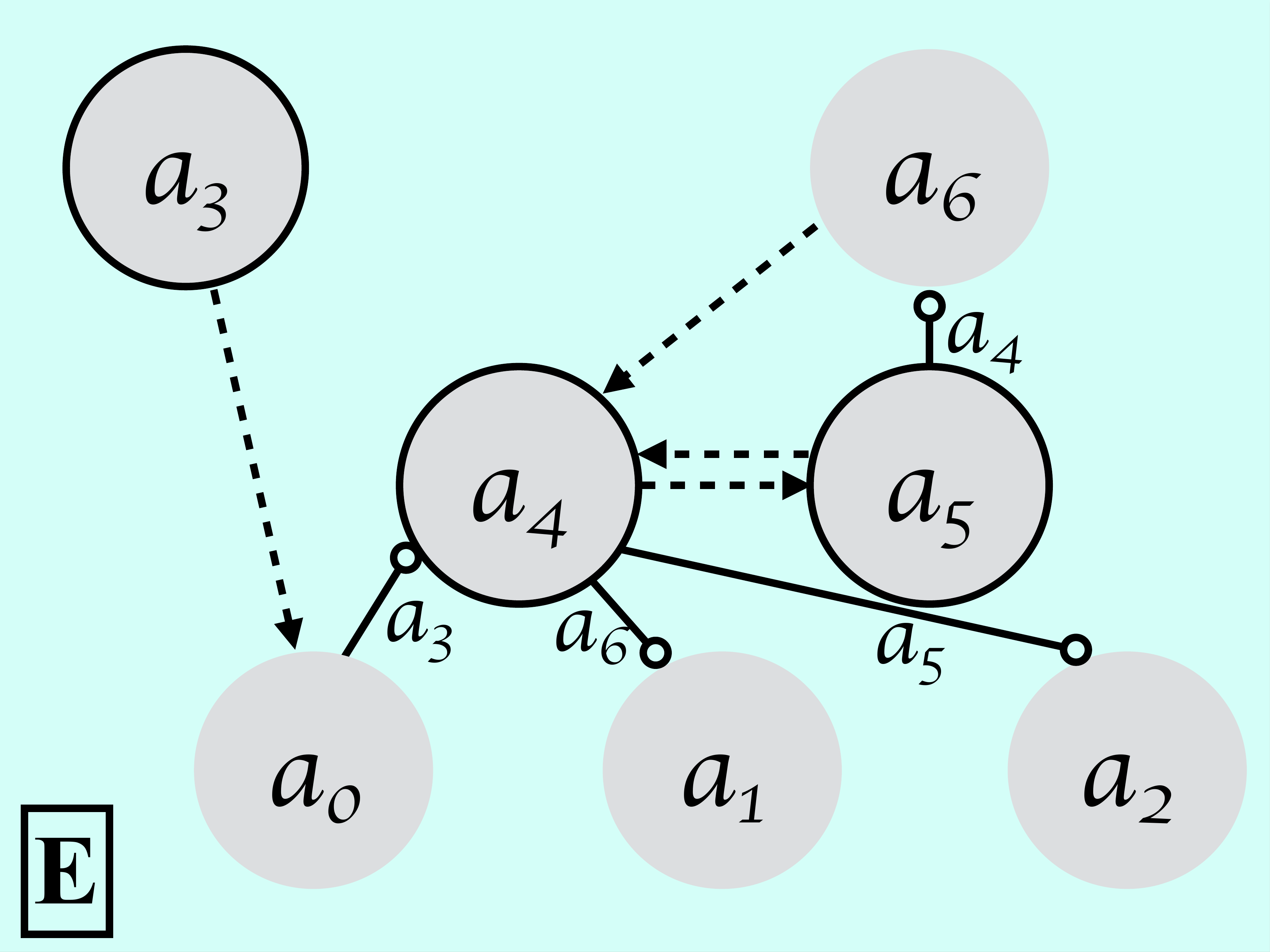}
    \includegraphics[scale=0.10]{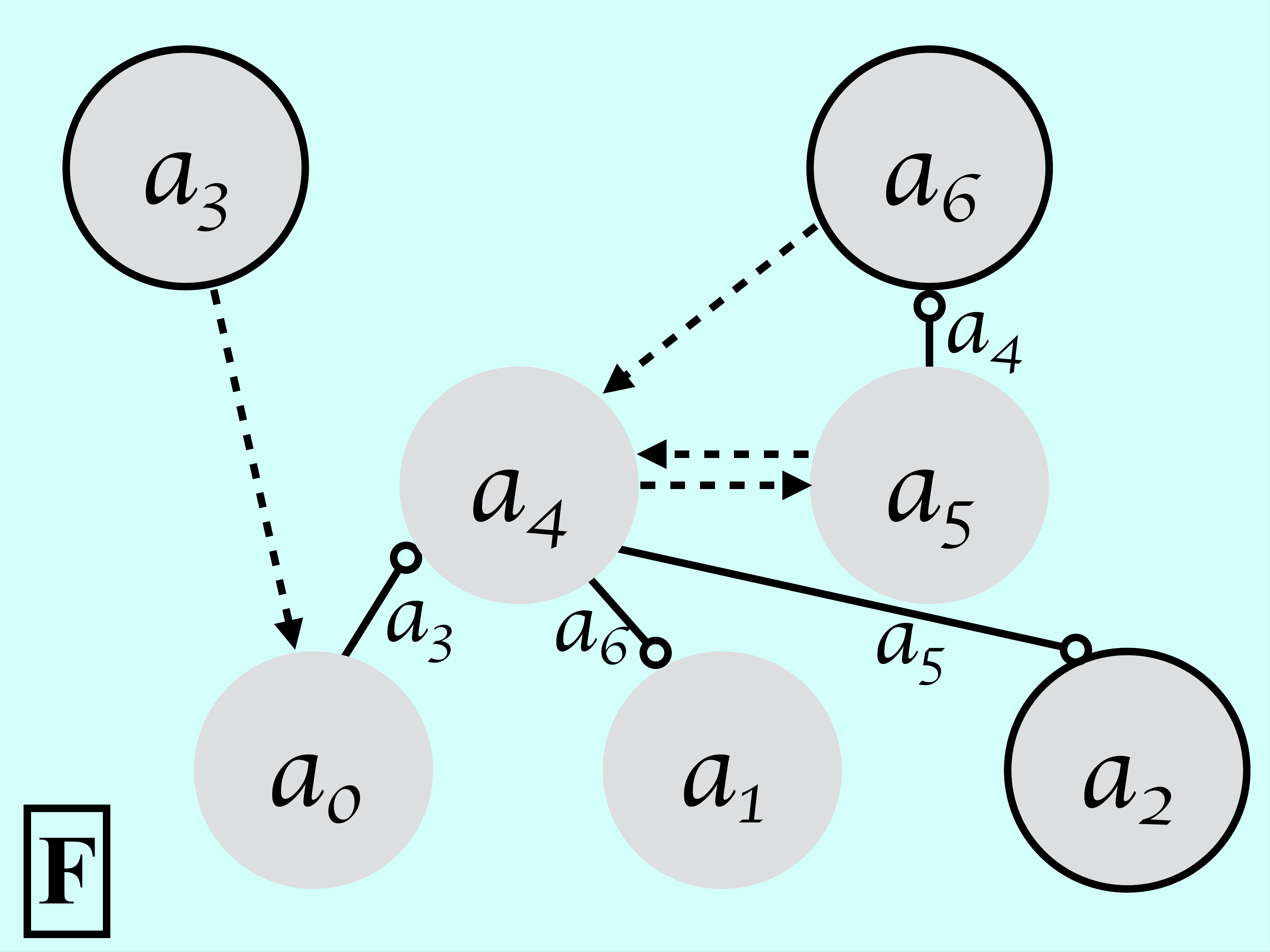}
\end{center}

      Assume $x \equiv (q_2, 1, q_3, q_4)$, 
       $a_0 = \sigma^Q(q_2), a_1 = \sigma^Q(q_3)$, $a_2 = \sigma^Q(q_4)$, $a_3 = \sigma^{\mathcal{I}}(x)$, 
       $a_4 = \sigma^{\mathcal{I}c}(x)$,  
       $a_5 = \sigma^1(n_1)$, and $a_6 = \sigma^1(n_1 - 1)$. $n \not= 0$, and in \fbox{D}, 
       $n$ is in fact assumed to be 1, so that $a_6 = \sigma^1(0)$, purely due to the space available in figure. 
       $a_0$, $a_1$ and $a_2$ represent states of two-counter Minsky machine, 
       $a_5$ and $a_6$ represent the first counter's content, 
       while $a_4$ is an auxiliary operational $\APA$ argument. There exist a sequence of $\APA$
       transitions into \fbox{E}, and then into \fbox{F}, as required. Similar when $1 < n$. 
 \item[Case 4] $(q_1, n_1, n_2) \underbrace{\Rightarrow \cdots \Rightarrow}_j (q_{2}, n_3, n_4) 
   \Rightarrow (q_4, n_3, n_4 - 1)$: similar. 
 \item[Case 5] $(q_1, n_1, n_2) \underbrace{\Rightarrow \cdots \Rightarrow}_j (q_{2}, 0, n_4) 
   \Rightarrow (q_3, 0, n_4)$:  A relevant snippet is shown in \fbox{G}, where $a_6 = \sigma^1(0)$, 
   and there exist $\APA$ transitions into \fbox{H}, and then into \fbox{I}, as required. 
   \begin{center}
         \includegraphics[scale=0.10]{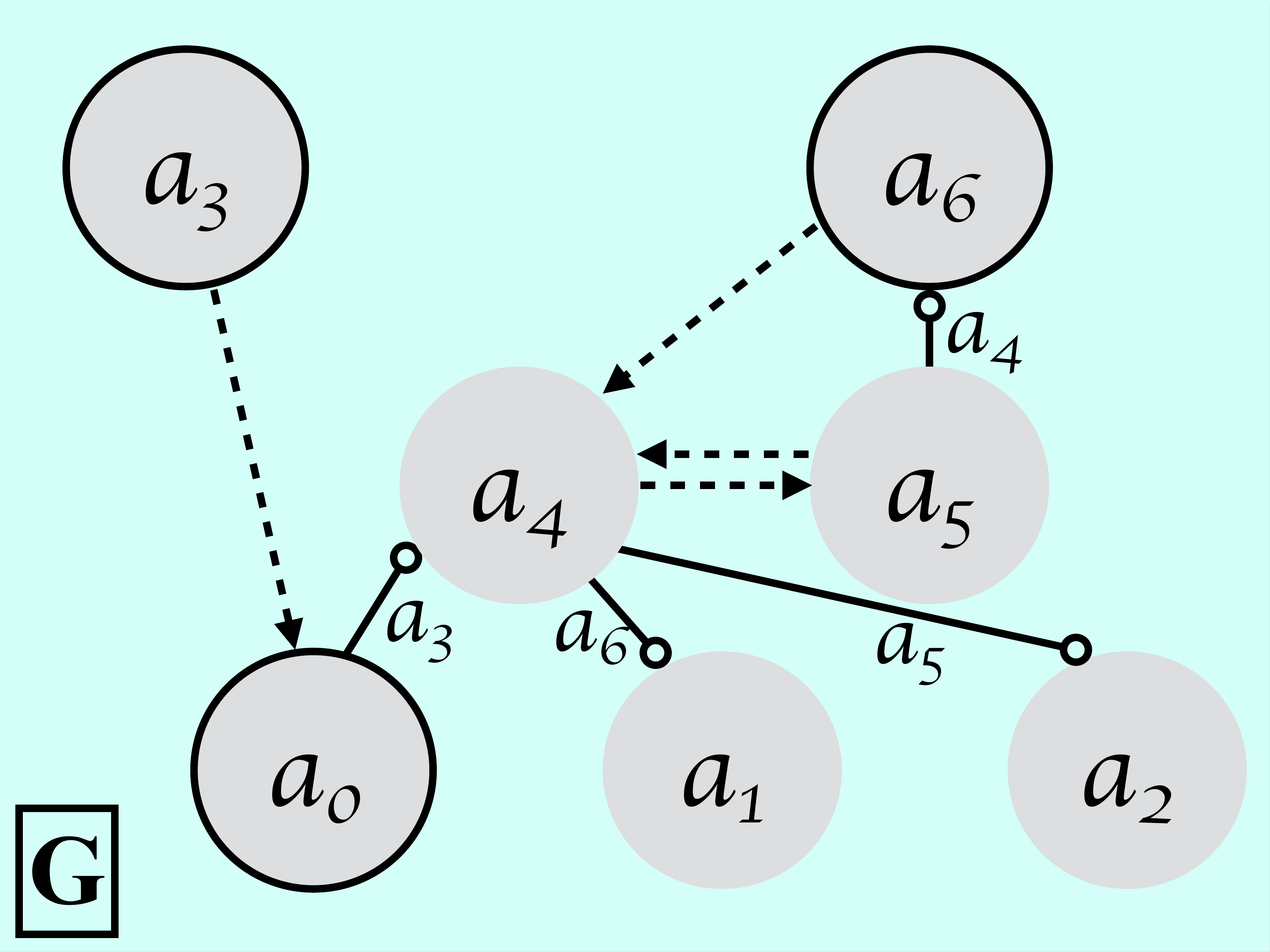}
    \includegraphics[scale=0.10]{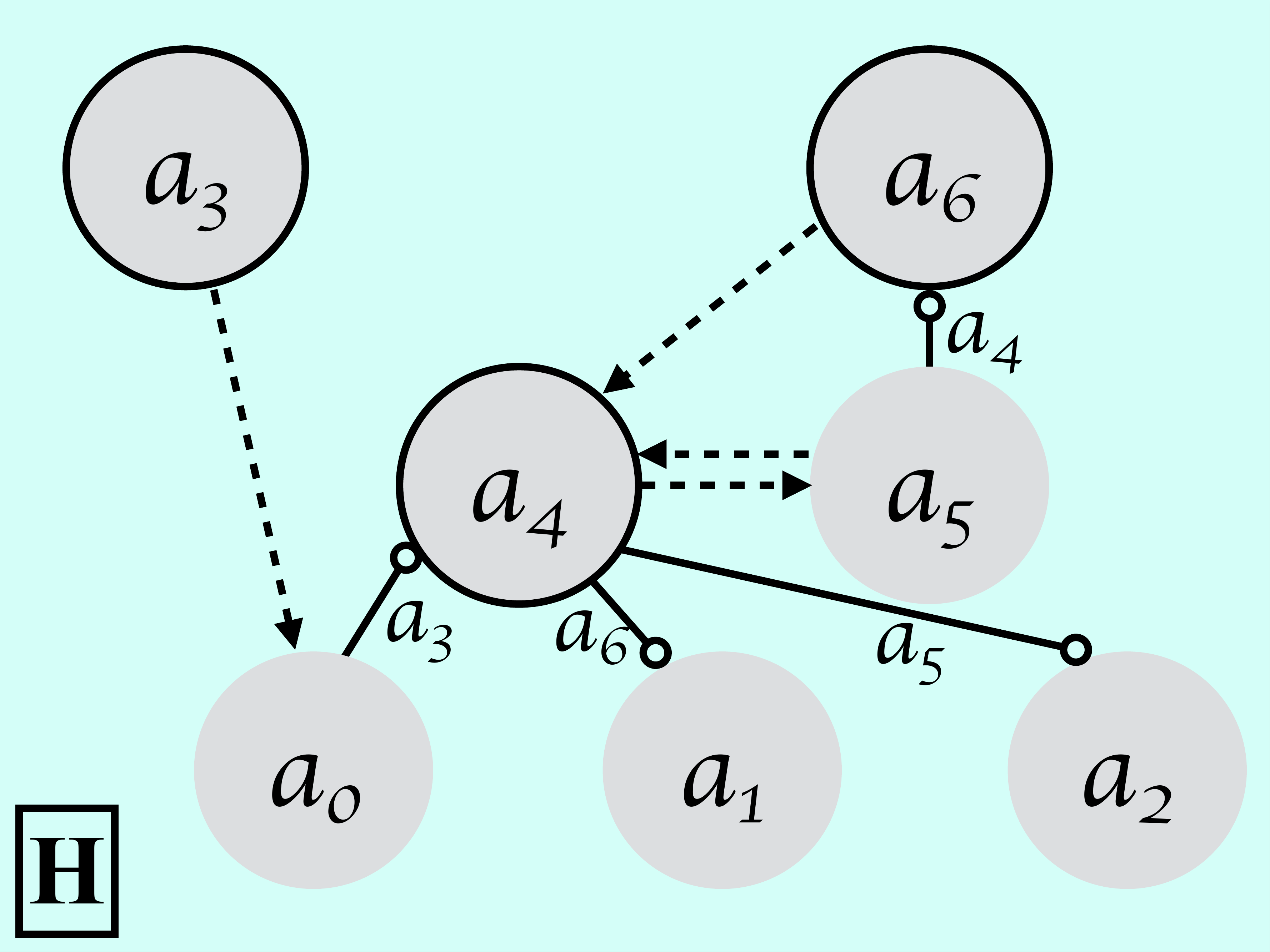}
    \includegraphics[scale=0.10]{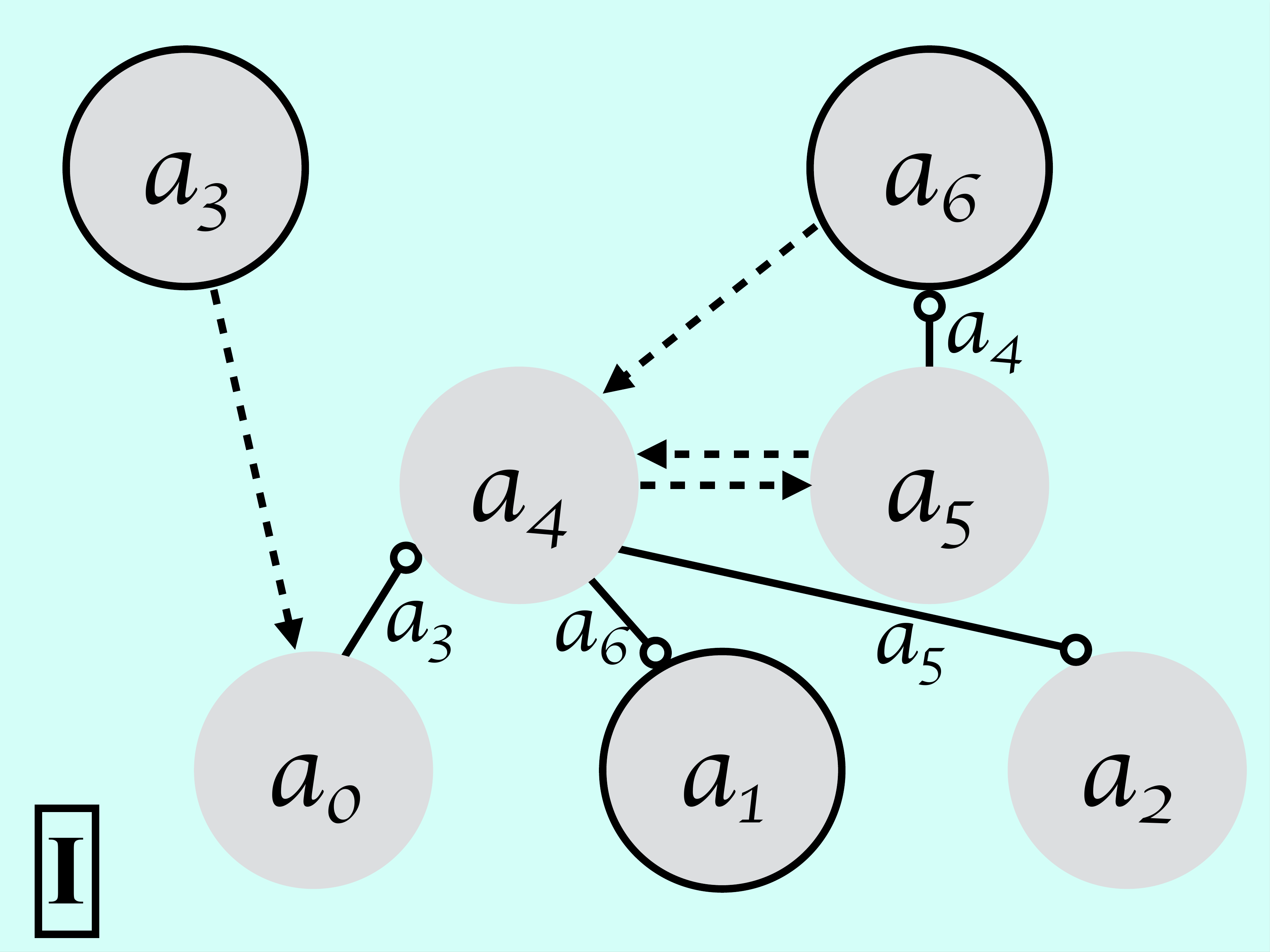}
   \end{center}
  \item[Case 6] $(q_1, n_1, n_2) \underbrace{\Rightarrow \cdots \Rightarrow}_j (q_{2}, n_3, 0) 
   \Rightarrow (q_3, n_3, 0)$: similar.   
   \qed 
\end{description}
\end{proof} 
\noindent \textbf{Conclusion.} We proved Turing-completeness of $\APA$ dynamics. 
\bibliography{references} 

\begin{thebibliography}{1}

\bibitem{ArisakaSatoh18}
R.~Arisaka and K.~Satoh.
\newblock {Abstract Argumentation / Persuasion / Dynamics}.
\newblock In {\em {PRIMA}}, pages 331--343, 2018.

\bibitem{Dung95}
P.~M. Dung.
\newblock On the {Acceptability} of {Arguments} and {Its} {Fundamental} {Role}
  in {Nonmonotonic} {Reasoning}, {Logic Programming}, and n-{Person} {Games}.
\newblock {\em Artificial {Intelligence}}, 77(2):321--357, 1995.

\bibitem{Minsky67}
M.~L. Minsky.
\newblock {\em Computation: Finite and Infinite Machines (Automatic
  Computation)}.
\newblock Prentice Hall, 5th edition, 1967.

\end{thebibliography}
   \bibliographystyle{abbrv}  
\end{document}